%% file: main.tex
\documentclass[twoside,legno,twocolumn,10pt]{article}
\usepackage{ltexpprt, fullpage,amsmath, amssymb, epsfig, algorithm, algorithmic, color}
\usepackage{graphicx,psfrag}
\usepackage{hyperref}

\renewcommand{\P}[2]{\mathbb{P}_{#1}\left(#2\right)}
\newcommand{\Pt}[1]{\widetilde{\mathbb{P}}\left(#1\right)}

\newcommand{\I}[1]{\mathbb{I}\left(#1\right)}
\newcommand{\E}[1]{\mathbb{E}\left[#1\right]}
\newcommand{\U}[2]{u_{#2}^{#1}}
\newcommand{\p}[2]{p_{#2}^{#1}}
\newcommand{\q}[2]{q_{#2}^{#1}}

\newcommand{\x}[2]{X_{#2}^{#1}}

\newcommand{\tm}[1]{\textrm{#1}}

\begin{document}

\title{\Large Prisoner's Dilemma on Graphs with Large Girth}

\author{Vahideh H. Manshadi
    \thanks{ Department of Electrical Engineering, Stanford University, Stanford, CA 94305.
    Email:\protect\url{vahidehh@stanford.edu}.}
\and
Amin Saberi\footnotemark[2]%
    \thanks{ Department of Management Science and Engineering, Stanford University, Stanford, CA 94305. Email:\protect\url{saberi@stanford.edu}.}
}

\date{}

\maketitle

\pagenumbering{arabic}
\setcounter{page}{1}

\input{abstrct}

%\setcounter{page}
%\thispagestyle{empty}
%\clearpage

\input{intro}
\input{definition}

\input{result}

\input{claim}

\input{uppr}

\input{lowr}
\input{example}
\input{discussion}

\vfill\eject
\bibliographystyle{abbrv}
\bibliography{references}

\end{document}

%% file: abstrct.tex
\begin{abstract}

We study the evolution of cooperation in populations where individuals play prisoner's dilemma on a network. Every node of the network corresponds on an individual choosing whether to cooperate or defect in a repeated game. The players revise their actions by imitating those neighbors who have higher payoffs.

We show that when the interactions take place on graphs with large girth, cooperation is more likely to emerge. On the flip side, in graphs with many cycles of length $3$ and $4$, defection spreads more rapidly.

One of the key ideas of our analysis is that our dynamics can be seen as a perturbation of the voter model. We write the transition kernel of the corresponding Markov chain in terms of the pairwise correlations in the voter model. We analyze the pairwise correlation and show that in graphs with relatively large girth, cooperators cluster and help each other.

\end{abstract}

%% file: intro.tex
\section{Introduction}
\label{sec:intro}
Prisoner's dilemma has become a known benchmark for studying the emergence of cooperation
in populations consisting of selfish agents. In this symmetric $2$-person game, each player has
two strategies, cooperate (C) or defect (D). A cooperator pays a cost of
$c$, and it provides the benefit of $b$ to the opponent. A defector incurs no cost and contributes no
benefit. If the game is only played once, basic analysis of the static game shows that
the unique Nash equilibrium is the defect-defect strategy.

Despite this prediction, the evolution of cooperation has been observed in
populations such as genomes, multicellular organisms, and human society. Such an abundance of cooperation in settings similar to prisoner's dilemma  has motivated
an extensive literature in game theory and evolutionary biology to explain the emergence of cooperation.  For example, for the two player repeated game,
the folk theorem implies that if the players are patient enough the cooperate-cooperate outcome is an equilibrium path of the
infinite horizon game. This is also extended to multi-player games, games of incomplete information, and noisy repeated games \cite{ellison, fundenberg3, fundenberg2}. This question is also studied in evolutionary game theory and evolutionary biology \cite{nowak,nowak_lieberman,fundenberg,GamesOnGraphs}.

In this paper, we study the evolution of cooperation in populations where each agent only interacts
with a small part of the population. In particular, individuals play
prisoner's dilemma on a network. Every node of the network corresponds on an individual
choosing whether to cooperate or defect in a repeated game.  The payoff of a node increases
with the number of its cooperator neighbors.  The network structure is also crucial for the evolution of play:
agents revise their actions by imitating those neighbors who have higher payoffs.  This is similar to the class of {\em Imitation
of Success} dynamics studied in evolutionary game theory (see e.g. Sandholm \cite{sandholm}).

We give the first rigorous analysis proving that, defined properly, {\em locality of interactions increases the likelihood of
the emergence of cooperation}. In particular, if the underling network {\em does not have any short cycles}, the expected number
of cooperators eventually exceeds its initial value. At the same time, we discover graphs with many cycles of
length $3$ or $4$, in which cooperation tends to decrease because of a ``free-riding'' effect.  In these graphs, the set of cooperators is always
surrounded by a set of defectors. Since defectors pay no cost, when they are well-connected
to the set of cooperators, their payoff will be large. Hence the probability of
imitating the defect action will increase.

On a more technical side, our key idea in analyzing this dynamics is that it can be viewed as a {\em perturbation} of the {\em Voter Model (VM)} \cite{aldous_fill,liggett}. We write the transition kernel of the Markov chain corresponding to our dynamics in terms of the pairwise correlations in the voter model.
We analyze the pairwise correlations and show that in graphs with relatively large girth,
local clustering occurs. In other words, we show that when the girth of the graph is relatively large, the cooperators will cluster together and help each other.
We also upper-bound the convergence time of
our dynamics using techniques similar to those used to bound the convergence
time of the voter model \cite{aldous_fill,Welsh}.

The dynamics studied in our paper is closely related to that of \cite{nowak,nowak_lieberman}. However, our approach differs in two
essential ways: first, unlike \cite{nowak,nowak_lieberman}, our results are rigorously proved. More importantly, \cite{nowak_lieberman},
focuses on the average degree of nodes, and ignores the role of small cycles. In our examples, we give graphs with short cycles that
do not show the behavior predicted in \cite{nowak_lieberman}.

%% file: definition.tex
\section{Problem Definition and Preliminaries}
\label{sec:definition}

We are given a set of players $V$, with $|V| = n$, that play prisoner's dilemma repeatedly at time
steps $t = 0, 1, 2,3, \dots$. The players interact on an undirected $k$-regular connected graph $G = (V,E)$
; each player only plays with its neighbors. Each player $i \in V$ has two alternative
strategies: cooperate (C) or defect (D). We associate variable $\x{i}{t} \in \{0,1\}$ to each node $i$ to represent its action
at time $t$; $\x{i}{t} =1$ ($\x{i}{t} = 0$) if player $i$ cooperates (defects) at time $t$. The vector,
$\vec{X}_t \in \{0,1\}^n$ represents the configuration of the system at time $t$. The payoff matrix
is a $2 \times 2$-matrix illustrated in Table \ref{tab:payoff}. Note that the game
is symmetric. The total payoff of player $i \in V$ at time $t$, $\U{i}{t}$, is:
\begin{eqnarray}
\U{i}{t} = -kc  \x{i}{t} + b  \sum_{j \sim i} \x{j}{t}
\label{eq:payoff}
\end{eqnarray}
where $j \sim i$ means that $j$ is a neighbor of $i$ in $G$. Further, we assume that $k(b+c) < 1$.

\begin{table}
\begin{center}
\begin{tabular}{ | l | l | l | }
\hline	
 $j$ $\backslash$ $i$ & C & D \\
\hline
 C & b-c & b \\
\hline
 D & -c & 0 \\
\hline
\end{tabular}
\end{center}
\caption{Payoff of player $i$}
\label{tab:payoff}
\end{table}

At each time step, one of the nodes is chosen uniformly at random to update its action. Let $\mathcal{A}_{t}^{i}$
be the event that at time $t$, node $i$ is chosen to update its action. If $\mathcal{A}_{t}^{i}$ occurs and $i$ alternates its
strategy then all the other players update their payoffs. Node $i$ updates its action according to the following mechanism:

\vspace{5mm}
{\noindent \textbf{Weak Imitation of Success (WIS):}} In the WIS dynamics, the updating node $i$ samples one of its neighbors and imitates its action. The sampling is slightly biased in favor of neighbors that have higher payoffs. In particular, node $i$ first samples a selector, $S_t$, that is an independent Bernouli($\epsilon$) random variable, where $\epsilon$ is a small positive number. If $S_t = 0$ then node $i$ samples one of its neighbors uniformly at random and
copies its action. On the other hand, if $S_t = 1$, node $i$ will take a biased sample among its neighbors, where $i$ favors neighbors
with higher payoffs. Formally,
\begin{align}
\P{}{i ~\tm{samples}~ j~|~\mathcal{A}_{t}^{i}, \vec{X}_t, \{S_t = 1\}} =  \quad \quad \quad \quad \nonumber \\\frac{1}{k} \left(\U{j}{t} +1 - \frac{1}{k}  \sum_{h \sim i} \U{h}{t} \right).
\label{eq:WIS:copy}
\end{align}

{\noindent It is easy to check that for $k(b+c) < 1$,}
\[
0 \leq \P{}{i ~\tm{samples}~ j~|~\mathcal{A}_{t}^{i}, \vec{X}_t, \{S_t = 1\}} \leq 1,
\]
and,
\[
\sum_{j \sim i} \P{}{i ~\tm{samples}~ j~|~\mathcal{A}_{t}^{i}, \vec{X}_t, \{S_t = 1\}} = 1.
\]
Putting these two cases together, the probability that node $i$ updates its action to $C$ is:

\begin{eqnarray}
 \P{}{\x{i}{t+1} = 1|\mathcal{A}_{t}^{i}, \vec{X}_t}  = \quad \quad \quad \quad \quad \quad \quad \quad \quad \nonumber \\ \frac{1}{k} \sum_{j \sim i}  \x{j}{t} \left[ 1 - \epsilon + \epsilon \left(\U{j}{t}+1 - \frac{1}{k}\sum_{h \sim i} \U{h}{t}\right)\right]  \label{eq:updat_WIS}
\end{eqnarray}

{\noindent It is worth noting that when $\epsilon = 0$, the dynamics coincides with the VM \cite{aldous_fill,liggett}. In the VM, the sampling of a neighbor is uniform and the updating node is more likely to imitate the strategy that the majority of its neighbors play, regardless of their payoffs. On the other hand, when $\epsilon = 1$, the sampling is based on the payoff of the PD games. We call this updating rule the PD dynamics. Note that for any $0< \epsilon <1$, the WIS is a mixture of VM and PD dynamics.}

It can be readily seen that WIS defines a Markov chain on $\{0,1\}^n$ with two absorbing states; the all zero state, $\vec{X} = \vec{0}$, and the all one state, $\vec{X} = \vec{1}$. Since the graph size is finite, starting from any configuration, the chain reaches
either $\vec{0}$ or $\vec{1}$ in a finite time. We denote the probability that the chain starting from configuration $\vec{X}_0$ converges to the all one state, $\vec{X} = \vec{1}$, by $\pi_{WIS,\vec{X}_0}$, i.e.,

\begin{eqnarray*}
%\label{eq:stationary}
\pi_{WIS,\vec{X}_0} = \lim_{t \rightarrow \infty} \mu_{t, \vec{X}_0}(\vec{X} = \vec{1})
\end{eqnarray*}

{\noindent where $ \mu_{t, \vec{X}_0}(\vec{X})$ is the measure defined by the Markov chain at time $t$ starting from configuration $\vec{X}_0$.
In this paper, we mainly work with one fixed initial condition, thus hereafter, we drop the subscripts $\vec{X}_0$.

%Note that for each $i \in V$, $\E{\lim_{t \rightarrow \infty} \x{i}{t}} = \pi_{WIS} $. Since random variable $\x{i}{t}$ is bounded, Fatou's lemma implies that:}
%\[
%\limsup_{t \rightarrow \infty} \E{\x{i}{t}} \leq \pi_{WIS} \leq \liminf_{t \rightarrow \infty} \E{\x{i}{t}}
%\]
%Since $\limsup$ of a sequence is at least equal to the $\liminf$, the above inequalities imply that
%the sequence $\E{\x{i}{t}}$ converges to $\pi_{WIS}$.
%Therefore, to compute the probability of convergence to the all one state, it suffices to analyze the evolution of the marginal probabilities.
%In particular, let $\p{i}{t}$ be the probability that player $i$ cooperates at time $t$.
%Then $\lim_{t \rightarrow \infty} \p{i}{t} = \pi_{WIS}$, for each $i \in V$. Also, note that the
%expected number of cooperators converges to $n \pi_{WIS}$.

%% file: result.tex
\section{Main Theorem}
\label{sec:result}
In this section, we present the main result of the paper. It states that
under the WIS dynamics, when interactions are local and the graph does not
have short cycles, and the benefit to cost ratio $b/c$ is moderately large,
then the expected number of cooperators increases.

\begin{theorem}
\label{thm:main}
Suppose graph $G$ is a connected $k$-regular graph with girth at least $7$. Further, suppose at time $0$, a random pair of neighbors play $C$ and the rest of the nodes play $D$; the system evolves according to the $WIS$ afterwards. For any $\gamma > 0$, $\epsilon = n^{-(4 + \gamma)}$, $b/c > \frac{k^2}{k-1}$, and $n$ sufficiently large, the probability that the chain converges to the all one state, $\vec{X} = \vec{1}$, is strictly larger than $2/n$. More precisely, there exists a positive constant $f$ that is bounded away from zero and
\begin{align*}
\pi_{WIS}  \geq \frac{2}{n}+ \frac{\epsilon}{n} f.
\end{align*}
The constant $f$ is increasing in the ratio $b/c$.
\end{theorem}

Initially, the number of cooperators in the network is $2$. Eventually, when the system
converges, the expected number of cooperators is $n \pi_{WIS}$. When $\epsilon = 0$, i.e., when we ignore the payoffs,
the expected number of cooperators stays at $2$. Our result shows that in our setting, when we slightly increase the $\epsilon$,
we strictly increase the value of $n \pi_{WIS}$. This is an indication that cooperation has a higher evolutionary fitness
on graphs with large girth.

Consider an updating
node $i$ and two of its neighbors $j$ and $h$.
Suppose $j$ cooperates and $h$ defects. If the set of neighbors of $j$ and $h$ are the same, which may happen
in graphs with cycles of length $4$, then the payoff of $h$ is always higher than payoff of
$j$ which results in the increase in the probability that $i$ samples $h$.
On other hand,
suppose that the set of neighbors of $j$ and $h$ are disjoint and suppose we can show that as a result of
clustering, in expectation, node $j$ has one more cooperator neighbor than $h$ does. In this case, the
difference between payoff of $j$ and payoff of $h$ is $b - kc$. Thus if $b/c> k$, then
the payoff of $j$ will be higher than the payoff of $h$ implying than $i$ is more likely
to sample $j$. In fact, the condition on the ratio of $b/c > k (1 + o(1))$ is needed even on an infinite tree.

\begin{figure*}
\centering
    \psfrag{t1}{\small{$0$}}
    \psfrag{t2}{\small{$1$}}
    \psfrag{t3}{\small{$2$}}
    \psfrag{t*}{\small{$t^{*}$}}
    \psfrag{t*1}{\small{$t^{*}+1$}}
    \psfrag{theta}{\small{$\theta$}}
    \psfrag{VM}{\small{$VM$}}
    \psfrag{PD}{\small{$PD$}}
    \psfrag{number}{\footnotesize{Expected $\#$ of Cooperators}}
    \psfrag{two}{\small{$2$}}
    \psfrag{towpdelta}{\small{$2+ \E{\Delta_{t^{*}}}$}}
	\includegraphics[height=.3\textheight,clip]{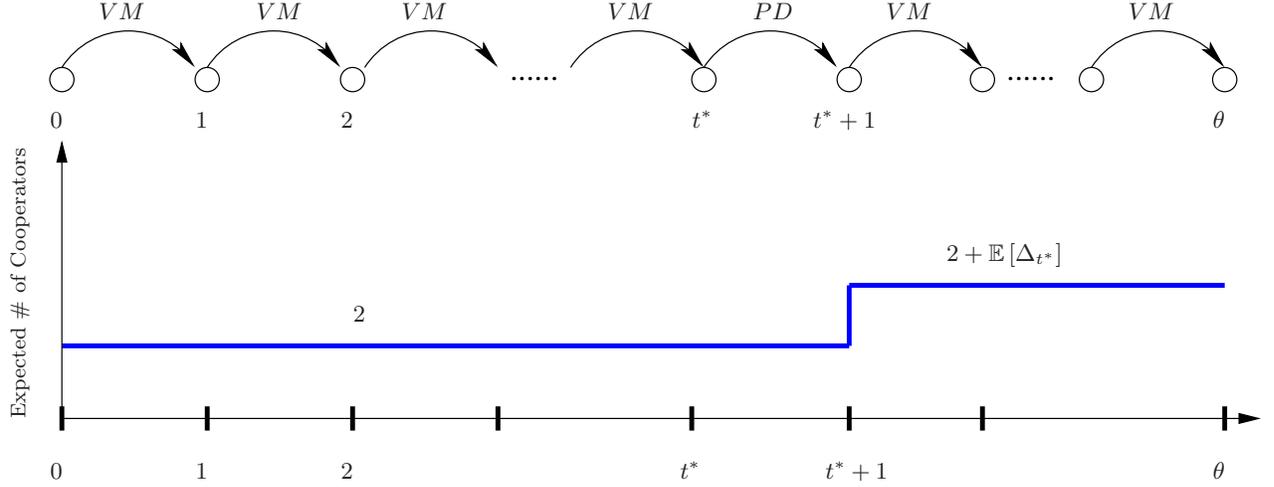}
\caption{Illustration of Claim \ref{thm:clm}; the expected number of cooperators increases when we apply the PD for the first time in step $t^{*}$.}
\label{fig:time}
\end{figure*}

Note that for each $i \in V$, $\E{\lim_{t \rightarrow \infty} \x{i}{t}} = \pi_{WIS} $. Since random variable $\x{i}{t}$ is bounded, Fatou's lemma implies that:}
\[
\limsup_{t \rightarrow \infty} \E{\x{i}{t}} \leq \pi_{WIS} \leq \liminf_{t \rightarrow \infty} \E{\x{i}{t}}
\]
Since $\limsup$ of a sequence is at least equal to the $\liminf$, the above inequalities imply that
the sequence $\E{\x{i}{t}}$ converges to $\pi_{WIS}$.
Therefore, to compute the probability of convergence to the all one state, it suffices to analyze the evolution of the marginal probabilities.
In particular, let $\p{i}{t}$ be the probability that player $i$ cooperates at time $t$.
Then for each $i \in V$, $\lim_{t \rightarrow \infty} \p{i}{t} = \pi_{WIS}$. Also, note that the
expected number of cooperators converges to $n \pi_{WIS}$. A major part of the proof of the
above theorem studies the evolution of the marginal probabilities and the expected number
of cooperators.

\vspace{2mm}
{\noindent \textbf{\normalsize{Main Proof Ideas:}}}
\vspace{2mm}

The key idea in the analysis is that the WIS is a mixture of two dynamics: the voter model (VM) and
the prisoner's dilemma (PD). At each time step, depending on the selector $S_t$, the updating player follows one of these
dynamics; if $S_t = 0$ the player follows the VM, otherwise it follows the PD.

One of the basic properties of the voter model on regular graphs is that given an arbitrary configuration,
at each time step $t$, the expected drift in the number of cooperators is zero: In a regular graph, for any edge
$(i,j)$, the probability that $i$ copies $j$ is the same as the
probability that $j$ copies $i$. Suppose at time $t$, $i$ plays $C$ and $j$ plays $D$. If $i$ copies $j$, the
number of cooperators decreases by one. On the other hand, if $j$ copies $i$, the
number of cooperators increases by one. Since the probability of these two events are the same, the expected drift
along each edge is zero, implying that the expected drift is zero regardless of the configuration at time $t$.

%
%Before proceeding with the analysis, we review some
%of the basic properties of VM in the following proposition. More interesting properties of the pairwise correlations
%in the voter model is studied in Subsection \ref{subsec:corr}.
%
%\begin{proposition}
%\label{prop:VM}
%Suppose we have a network of agents given by the graph $G$. Suppose the initial distribution of the actions is given by $\mu_0$ and the expected number of cooperators with respect to $\mu_0$ is $N_0$. If the system evolves according to the voter model, at any time $t$, the expected number of cooperators (with respect to the dynamics and the initial distribution) is $N_0$.
%\end{proposition}
%
%\begin{proof}
%Given an arbitrary configuration at any time $t$, we show that the expected drift in the number of cooperators is zero.
%Since this holds independent of the initial configuration, the same is true for any initial distribution on the
%configurations. In a regular graph, for any edge $(i,j)$, the probability that $i$ copies $j$ is the same as the
%probability that $j$ copies $i$. Suppose at time $t$, $i$ plays $C$ and $j$ plays $D$. If $i$ copies $j$, the
%number of cooperators decreases by one. On the other hand, if $j$ copies $i$, the
%number of cooperators increases by one. Since the probability of these two events are the same, the expected drift
%along each edge is zero, implying that the expected drift is zero regardless of the configuration at any time $t$.
%\end{proof}

Let $N_{\theta}$ denote the number of cooperators at time ${\theta}$, i.e., $N_{\theta} = \sum_{i \in V} \x{i}{\theta}$.
%Suppose the system converges by time $\theta$, i.e., $\theta$ is the earliest time that the system reaches either the all one ($\vec{X} = \vec{1}$) or the all zero ($\vec{X} = \vec{0}$) state. Clearly $\pi_{WIS} = \E{N_{\theta}}/n$.
First, we study $\E{N_{\theta}}$ and compute a lower-bound for it.
Remember that in the WIS, at each time step $t = 0,1, \ldots, \theta-1$, we first sample a selector, $S_{t}$, that is an independent Bernouli($\epsilon$) random variable, based on which we follow either VM or PD. Given the sequence $S_0, S_1, \ldots, S_{\theta -1}$, consider the following three cases:

\vspace{2mm}
{\noindent \textbf{Case 1: all the selectors $S_0, S_1, \ldots, S_{\theta -1}$ are zero:}}
\vspace{2mm}

{\noindent In this case the system has followed the VM process. Therefore, the expected number of cooperators at time $\theta$ remains $2$. }

\vspace{2mm}
{\noindent \textbf{Case 2: exactly one of the selectors $S_0, S_1, \ldots, S_{\theta -1}$ is one:}}
\vspace{2mm}

{\noindent Suppose the selector at time $t^*$ is one. Since the system has followed the VM up to time $t^*$, the expected number of cooperators at time $t^*$ is $2$. However, after applying the PD at time step $t^*$, the expected number of cooperators changes. We denote the expected drift by $\E{\Delta_{t^*}}$ (see Figure \ref{fig:time}). A key part of the analysis is to show that $\E{\Delta_{t^*}}$ is positive, i.e., applying the PD for one period of time results in an increase in the expected number of cooperators. More precisely,}

\begin{claim}
\label{thm:clm}
Given graph $G$, suppose at time $0$, a random pair of neighbors play $C$ and the rest of the  nodes play $D$; the system evolves according to the $WIS$ afterwards. If $b/c \geq k^2/(k-1)$, then $0 \leq \E{\Delta_{t^{*}}} \leq 1$.
\end{claim}

{\noindent The proof of this claim is presented in Section \ref{sec:clm}. From $t^*+1$ to $\theta$, the system again evolves according to the VM model, therefore in the time periods $t^*+1$ to $\theta$, the expected number of cooperators remains the same.}\\

\vspace{2mm}
{\noindent \textbf{Case 3: more than one of the selectors $S_0, S_1, \ldots, S_{\theta -1}$ are one:}}
\vspace{2mm}

{\noindent In this case, we just lower-bound the expected number of cooperators by zero.}

\vspace{3mm}

{\noindent Putting the three cases together, we have:}

\begin{align}
& \E{N_{\theta}}  \geq \E{N_{\theta} \I{ \bigwedge_{t = 0}^{\theta -1} \{S_{t} = 0\}}}
\nonumber \\
& + \sum_{t^{*} = 0}^{\theta -1} \E{N_{\theta} \I{ \bigwedge_{t = 0, t \neq t^{*}}^{\theta -1} \{S_{t} = 0\} \wedge \{S_{t^{*}} = 1\}}}
\nonumber \\
& = 2 (1 - \epsilon)^{\theta} + \epsilon (1 - \epsilon)^{\theta -1} \sum_{t^{*} = 0}^{\theta -1} \left(2 + \E{\Delta_{t^{*}}} \right).
\label{eq:thm:exp1}
\end{align}
{\noindent Define, }
\begin{align}
\delta_{\theta} = \frac{1}{\theta} \sum_{t^{*} = 0}^{\theta -1} \E{\Delta_{t^{*}}}
\label{eq:delta:def}
\end{align}
{\noindent We can rewrite \ref{eq:thm:exp1} as,}
\begin{align}
\E{N_{\theta}} \geq 2 (1 - \epsilon)^{\theta} + \left[2+ \delta_\theta \right] \theta \epsilon (1 - \epsilon)^{\theta},
\label{eq:thm:exp}
\end{align}
Function $g(\epsilon) = (1 - \epsilon)^{\theta}$ is convex for $\theta \geq 2$, therefore $(1 - \epsilon)^{\theta} \geq 1 - \epsilon \theta$.
Substituting this in \eqref{eq:thm:exp}, we have:
\begin{align}
\E{N_{\theta}} \geq 2 - \left[2 +  \delta_{\theta} \right]  \theta^2 \epsilon^2 + \theta \delta_{\theta}  \epsilon \geq 2 -3  \theta^2 \epsilon^2 + \theta \delta_{\theta}  \epsilon
\label{eq:thm:exp:simple}
\end{align}
where in the last inequality, we used the fact that $\delta_{\theta} < 1$, which follows from Claim \ref{thm:clm}.
Let $T^{*}$ be $0.5(k+1)n^{2+ \gamma/3}$, and let $\mathcal{C}_{T^{*}}$ be the event that the chain converges by time $T^{*}$,
we have:
\begin{align}
\E{N_{T^{*}}} & = \E{N_{T^{*}} |~\mathcal{C}_{T^{*}}} \P{}{\mathcal{C}_{T^{*}}} + \E{N_{T^{*}} |~\overline{\mathcal{C}}_{T^{*}}} \P{}{\overline{\mathcal{C}}_{T^{*}}} \nonumber \\
& \leq \E{N_{T^{*}} |~ \mathcal{C}_{T^{*}}} \P{}{\mathcal{C}_{T^{*}}} + n \P{}{\overline{\mathcal{C}}_{T^{*}}}
\label{eq:thm:Tstar}
\end{align}
where in the last inequality we use that $N_{T^{*}} \leq n$. Note that $\pi_{WIS} \geq \frac{1}{n}\E{N_{T^{*}} | ~\mathcal{C}_{T^{*}}} \P{}{\mathcal{C}_{T^{*}}}$. Thus putting \eqref{eq:thm:exp:simple} and \eqref{eq:thm:Tstar} together, we have:
\begin{align}
\pi_{WIS} \geq \frac{1}{n} \left[2-3  {T^{*}}^2 \epsilon^2 + T^{*} \delta_{T^{*}}  \epsilon -n \P{}{\overline{\mathcal{C}}_{T^{*}}}\right]
\label{eq:thm:last}
\end{align}

%
%Taking expectation with respect to $\theta$, we have:
%\begin{align}
%\mathbb{E}_{\theta} \left[\E{N_{\theta}} \right] \geq 2 -3  \E{\theta^2} \epsilon^2 + \E{\theta \delta_{\theta} } \epsilon
%\label{eq:thm:exp_theta}
%\end{align}

Having inequality \eqref{eq:thm:last}, the rest of the proof consists of establishing an upper-bound for $\P{}{\overline{\mathcal{C}}_{T^{*}}}$ in Lemma \ref{lem:uppr} and a lower-bound for $T^{*} \delta_{T^{*}}$ in Lemma \ref{lem:lowr}.

\begin{lemma}
\label{lem:uppr}
Given graph $G$, suppose the system evolves according to the $WIS$. For $\epsilon < n^{-4}$, we have $\P{}{\overline{\mathcal{C}}_{T^{*}}} \leq \frac{1}{2^{n^{\gamma/3}}}$.
\end{lemma}

We prove this lemma by mapping our Markov chain into a death-birth chain and using some results on death-birth chain and the basic voter model. The proof is presented in Section \ref{sec:uppr}.

%\begin{lemma}
%\label{lem:lowr}
%Suppose at time $0$, a random pair of neighbors play $C$ and the rest of the nodes play $D$; the system evolves according to the $WIS$ afterwards. For $\delta_{\theta}$ defined as in \ref{eq:delta:def},
%\begin{align*}
%\E{\theta \delta_{\theta} } \geq \e{2} (n-2) \delta_{n-2}
%\end{align*}
%\end{lemma}
%The proof of this lemma is presented in Appendix \ref{sec:lowr}. Finally, in the next lemma, we prove a stronger lower bound for
%$(n-2) \delta_{n-2} = \sum_{t* =0}^{n-3} \E{\Delta_{t^*}}$, which leads us to compute the constant $f$.

\begin{lemma}
\label{lem:lowr}
Given graph $G$, suppose at time $0$, a random pair of neighbors play $C$ and the rest of the nodes play $D$;
the system evolves according to the $WIS$ afterwards. If $b/c > \frac{k^2}{k-1}$, then for $n$ sufficiently large:
\begin{align*}
T^{*} \delta_{T^{*}} \geq \frac{k (k-1)^3 \sqrt{(k-1)/2}}{k^2(k-1) \sqrt{(k-1)/2} + k} (b - \frac{k^2}{k-1} c)
\end{align*}
\end{lemma}
The proof of this lemma is presented in Section \ref{sec:lowr}.

%% file: claim.tex
\section{Proof of Claim \ref{thm:clm}}
\label{sec:clm}
In  this section, we prove Claim \ref{thm:clm} by studying the expected drift in the number of cooperators when we apply the PD dynamics for the first time at $t^{*}$. Suppose node $i$ is chosen to update its action at time $t^{*}$. Substituting the payoff function \eqref{eq:payoff} into the
sampling probability \eqref{eq:WIS:copy}, we have:
%\begin{align*}
%& \P{}{\x{i}{t^{*}+1} = 1|~\mathcal{A}_{t^*}^{i}, \vec{X}_{t^*}, \{S_{t^* =1}\}} =\nonumber \\ & \quad \frac{1}{k} \sum_{j \sim i} \x{j}{t^*} \left[\frac{k-1}{k} \left(-kc  \x{j}{t^*}  \right. \right. \\
%& \quad \quad \quad \quad \quad \quad \quad \quad \quad \quad \left. \left.  + b  \sum_{l \sim j} \x{l}{t^*}\right) +1 \right]
%\\  & \quad -  \frac{1}{k} \sum_{j \sim i}  \x{j}{t^*}\left[\frac{1}{k}\sum_{h \sim i, h \neq j} \left(-kc  \x{h}{t^*} \right. \right. \\
%& \quad \quad \quad \quad \quad \quad \quad \quad \quad \quad \quad \quad \left. \left. + b  \sum_{g \sim h} \x{g}{t^*} \right)\right]
%\end{align*}

\begin{align*}
& \P{}{\x{i}{t^{*}+1} = 1|~\mathcal{A}_{t^*}^{i}, \vec{X}_{t^*}, \{S_{t^* =1}\}} =\nonumber \\ & \quad \frac{1}{k} \sum_{j \sim i} \x{j}{t^*} \left[\frac{k-1}{k} \left(-kc  \x{j}{t^*} + b  \sum_{l \sim j} \x{l}{t^*}\right) +1 \right]
\\  & -  \frac{1}{k} \sum_{j \sim i}  \x{j}{t^*}\left[\frac{1}{k}\sum_{h \sim i, h \neq j} \left(-kc  \x{h}{t^*} + b  \sum_{g \sim h} \x{g}{t^*} \right)\right]
\end{align*}

{\noindent Simplifying the above expression results in:}
%\begin{align*}
%& \P{}{\x{i}{t^{*}+1} = 1|~\mathcal{A}_{t^*}^{i}, \vec{X}_{t^*}, \{S_{t^* =1}\}} \\
%& = \frac{1}{k} \sum_{j \sim i}  \x{j}{t^*}   \\
%&+ \frac{k-1}{k^2} \sum_{j \sim i} \left[  -kc  \x{j}{t^*}  \right. \\
%& \quad \quad \quad \quad \quad \quad \left. +  b  \sum_{l \sim j, l \neq i} \x{l}{t^*}\x{j}{t^*}   \right] \\
%& - \frac{1}{k^2} \sum_{j \sim i}  \sum_{h \sim i, h \neq j} \left[ -kc  \x{h}{t^*} \x{j}{t^*}  \right. \\
%& \quad \quad \quad \quad \quad \quad \left. + b  \sum_{g \sim h, g \neq i} \x{g}{t^*}\x{j}{t^*} \right]
%\end{align*}

\begin{align*}
& \P{}{\x{i}{t^{*}+1} = 1|~\mathcal{A}_{t^*}^{i}, \vec{X}_{t^*}, \{S_{t^* =1}\}}  = \frac{1}{k} \sum_{j \sim i}  \x{j}{t^*}   \\
&+ \frac{k-1}{k^2} \sum_{j \sim i} \left[  -kc  \x{j}{t^*}  +  b  \sum_{l \sim j, l \neq i} \x{l}{t^*}\x{j}{t^*}   \right] \\
& - \frac{1}{k^2} \sum_{j \sim i}  \sum_{h \sim i, h \neq j} \left[ -kc  \x{h}{t^*} \x{j}{t^*} + b  \sum_{g \sim h, g \neq i} \x{g}{t^*}\x{j}{t^*} \right]
\end{align*}

{\noindent Note that the cooperation probability of node $i$ at time $t^*+1$ does not depend on its own action at time $t^{*}$. Further, it is decreasing in the number cooperator neighbors, but it is increasing in the number of {\em pairs} of neighbors playing $C$ and also in the number of
$C-C$ edges that are incident to neighbors of $i$.}

%{\noindent By linearity of expectations, we have:}
%\begin{align*}
%\E{N_{t^*+1}| ~\vec{X}_{t^*}, \{S_{t^* =1}\}} & = \sum_{i \in V} \P{}{\x{i}{t^{*}+1} = 1|\vec{X}_{t^*}, \{S_{t^* =1}\}} \\
%& = \frac{1}{n}\sum_{i \in V} \P{}{\x{i}{t^{*}+1} = 1|\mathcal{A}_{t^*}^{i}, \vec{X}_{t^*}, \{S_{t^* =1}\}} +
%\frac{n-1}{n}N_{t^*} \\
%& = N_{t^*} + \frac{1}{ n k^2} \sum_{i \in V} \left[ (k-1) \sum_{j \sim i} \left(-kc  \I{\x{j}{t^*} = 1} + b  \sum_{l \sim j, l \neq i} \I{\x{l}{t^*} = 1,\x{j}{t^*} = 1}\right) \right.\\
%& \left. - \sum_{j \sim i}  \sum_{h \sim i, h \neq j} \left(-kc  \I{\x{h}{t^*} = 1,\x{j}{t^*} = 1} + b  \sum_{g \sim h, g \neq i} \I{\x{g}{t^*} = 1, \x{j}{t^*} = 1}\right)  \right]
%\end{align*}

{\noindent Given $\vec{X}_{t^*}$, the expected drift in the number of cooperators, $\E{\Delta_{t^{*}}}$ in Figure \ref{fig:time}, is,
\begin{align}
& \E{\Delta_t^*| ~\vec{X}_{t^*}}  = \nonumber \\
& \frac{1}{nk}\sum_{i \in V} \left[ - (k-1)c\sum_{j \sim i} \x{j}{t^*} \right. \nonumber \\
& \quad \quad \quad \left.+ \frac{(k-1)b}{k} \sum_{j \sim i} \sum_{l \sim j, l \neq i}  \x{l}{t^*} \x{j}{t^*}  \right. \nonumber \\
& \quad \quad \quad \left.+ c\sum_{j \sim i}  \sum_{h \sim i, h \neq j} \x{h}{t^*} \x{j}{t^*}  \right. \nonumber \\
& \quad \quad \quad \left.- \frac{b}{k}  \sum_{j \sim i}  \sum_{h \sim i, h \neq j} \sum_{g \sim h, g \neq i} \x{g}{t^*} \x{j}{t^*}  \right]
\label{eq:Delta:sum:config}
\end{align}
{ \noindent Taking expectation over all configurations, $\vec{X}_{t^*} \in \{0,1\}^{n}$, we have:
\begin{align}
& \E{\Delta_{t^*}} =   \nonumber \\
&\frac{1}{nk} \sum_{i \in V} \left[- (k-1)c \sum_{j \sim i} \p{j}{t^*}+ \frac{(k-1)b}{k} \sum_{j \sim i} \sum_{l \sim j, l \neq i} \p{lj}{t^*}
\right. \nonumber \\
& \left.+ c \sum_{j \sim i}  \sum_{h \sim i, h \neq j} \p{hj}{t^*}- \frac{b}{k}  \sum_{j \sim i}  \sum_{h \sim i, h \neq j} \sum_{g \sim h, g \neq i} \p{gj}{t^*} \right]
\label{eq:Delta:sum}
\end{align}
where $\p{ij}{t^*} = \E{\x{i}{t^*} \x{j}{t^*}}$ is the pairwise correlation of $i$ and $j$ at time $t^*$, which is equal to the joint probability that nodes $i$ and $j$ play $C$ at time $t^*$. Since up until time $t^*$, the system evolves according to the VM dynamics, we analyze the drift $\E{\Delta_{t^*}}$ by studying the evolution of the pairwise correlations in the voter model.}

Before proceeding with the analysis, to simplify the representation, we introduce the {\em average correlations over all pairs} in graph $G$ and rewrite the drift $\E{\Delta_{t^*}}$ in terms of these average correlations. For a graph with girth at least $7$, and for $d \leq 4$, let $\q{(d)}{t}$ be the average joint probability of the two endpoints of any path of length $d$, i.e.,

{\noindent }
\begin{align}
\q{(d)}{t} = \frac{2}{nk(k-1)^{d-1}} \sum_{l \in L(d)} \p{u_lv_l}{t}, ~~~ 1 \leq d \leq 4
\label{eq:def:q}
\end{align}
where $L(d)$ is the set of all length $d$ paths, and $u_l$ and $v_l$ are the two endpoints of path $l$. Similarly, define
\begin{align}
\q{(0)}{t} = \frac{1}{n} \sum_{i \in V} \p{i}{t}.
\label{eq:def:q}
\end{align}
For a graph with girth at least $7$, between any pair of nodes that appear in the RHS of \eqref{eq:Delta:sum}, there is a unique path of length at most $3$. Thus, by rearranging the sums in \eqref{eq:Delta:sum} and using definition \eqref{eq:def:q}, we can write $\E{\Delta_{t^*}}$ in terms of the average correlations:

\begin{align}
& \E{\Delta_{t^*}} = \nonumber \\
& {(k-1)c} \left[-\q{(0)}{t^{*}} + \frac{(k-1)b}{kc} \q{(1)}{t^{*}} + \q{(2)}{t^{*}} - \frac{(k-1) b}{kc} \q{(3)}{t^{*}} \right].
\label{eq:Delta:Ave}
\end{align}

\subsection{Pairwise Correlations in the Voter Model}
\label{subsec:corr}
Consider the nodes $i$ and $j$. For $t \leq t^*$, we can write $\p{ij}{t}$ in terms of the pairwise joint probabilities of the previous step in the following way: at time $t-1$, if neither $i$ nor $j$ is chosen to update its action, then the joint probability of $i$ and $j$ does not change. On the other hand, if $i$ is chosen to update, then it will copy neighbor $h$ with probability $1/k$. Therefore, the joint probability of $i$ and $j$ at time $t$ will be the same as the joint probability of $h$ and $j$ in the previous step (time $t-1$). Similarly, if $j$ is chosen to update, we can write the joint probability of $i$ and $j$ in terms of the joint probability of $i$ and neighbors of $j$ at time $t-1$:

%\begin{align}
%\p{ij}{t}  & = (1 - \frac{2}{n}) \p{ij}{t-1}  + \frac{1}{nk} \sum_{h \sim i} \p{jh}{t-1} + \frac{1}{nk} \sum_{l \sim j} \p{il}{t-1}
%\nonumber \\  & \quad \quad \quad \quad \quad \quad \quad \quad \quad \quad \quad \quad ,t \leq t^* \label{eq:corr:recurs}
%\end{align}

\begin{align}
\p{ij}{t}  & = (1 - \frac{2}{n}) \p{ij}{t-1}  + \frac{1}{nk} \sum_{h \sim i} \p{jh}{t-1} + \frac{1}{nk} \sum_{l \sim j} \p{il}{t-1}
,~t \leq t^* \label{eq:corr:recurs}
\end{align}

%Since at time $0$, a random pair of neighbors play $C$, we have $\p{ij}{0} = 2/(nk)$, for $j \sim i$, and $\p{ij}{0} = 0$, for $j \nsim i$. For these natural initial conditions, it is easy to show that the pairwise correlations is nonincreasing in distance, i.e., if $d(i,j) \leq d(i,l)$, then $\p{ij}{t} \geq \p{il}{t}$, for $0 \leq t \leq t^*$.

In a graph with girth not smaller than $7$, for each node $i$, the subgraph that consists of all the nodes at distance at most $3$ from $i$ is a $k$-regular tree. Further, for $1 \leq d \leq 3$, the average correlation $\q{(d)}{t}$ is exactly the probability that at time $t$, a randomly selected pair of nodes at distance $d$ paly $C$. Thus for $1 \leq d \leq 3$ and $t \leq t^*$, the evolution of $\q{(d)}{t}$ is similar to the evolution of the pairwise correlation in an infinite $k$-regular tree. In particular, using \eqref{eq:corr:recurs} and \eqref{eq:def:q}, for $1 \leq d \leq 3$ and $t \leq t^*$, we have:}

\begin{align}
\q{(d)}{t} = (1-\frac{2}{n}) \q{(d)}{t-1} + \frac{2}{nk} \q{(d-1)}{t-1} + \frac{2(k-1)}{nk} \q{(d+1)}{t-1},
\label{eq:corr:ave}
\end{align}

Since at time $0$, a random pair of neighbors play $C$, we have $\q{(0)}{0} = 2/n$, $\q{(1)}{0} = 2/(nk)$, and $\q{(d)}{0} = 0$, for $d \geq2 $. Starting from these natural initial conditions, a simple induction shows that:

\begin{align}
\q{(4)}{t}  \leq \q{(3)}{t} \leq \q{(2)}{t} \leq \q{(1)}{t}, ~~~~ t \leq t^*.
\label{eq:corr:monotone}
\end{align}

%{\noindent Also, note that for graphs with girth at least $7$, if there is a path of length $4$ between nods $u$ and $v$ then distance between $u$ and $v$ is at least $3$. Thus, }
%\[
%\q{(4)}{t} \leq \q{(d)}{t}, ~~~~ d \leq 3, ~t \leq t^*.
%\]

This monotonicity implies that $\E{\Delta_{t^*}}$ is non-decreasing in the ratio $b/c$. Thus for $b/c \geq k^2/(k-1)$, we have:
\begin{align*}
\E{\Delta_{t^*}} \geq {(k-1)c} \left[-\q{(0)}{t^{*}} + k \q{(1)}{t^{*}} + \q{(2)}{t^{*}} - k \q{(3)}{t^{*}} \right]
\end{align*}

%{\noindent Further, using \eqref{eq:corr:recurs} and \eqref{eq:def:q}, it is not hard to show that for $1 \leq d \leq 3$ and $t \leq t^*$:}
%
%\begin{align}
%\q{(d)}{t} = (1-\frac{2}{n}) \q{(d)}{t-1} + \frac{2}{nk} \q{(d-1)}{t-1} + \frac{2(k-1)}{nk} \q{(d+1)}{t-1},
%\label{eq:corr:ave}
%\end{align}

{\noindent Using the recursive relation \eqref{eq:corr:ave}, we rewrite $\E{\Delta_{t^*}}$ as the following geometric sum:}
\begin{align}
& \E{\Delta_{t^*}} \geq \nonumber \\
& {(k-1)c} (1 - \frac{2}{n})^{t^*} \left[-\q{(0)}{0} + k \q{(1)}{0} + \q{(2)}{0} - k \q{(3)}{0} \right]  \nonumber \\
& + {(k-1)c}  \left[\frac{2}{n}  \sum_{s = 0}^{t^*-1} (1- \frac{2}{n})^{t^*-s-1} \left( \frac{1}{k} \q{(1)}{s} + (k-2) \q{(2)}{s} \right. \right. \nonumber \\
& \left. \left.  + \frac{k-1}{k} \q{(3)}{s} - (k-1) \q{(4)}{s}\right)\right]
\label{eq:Delta:corr}
\end{align}

{\noindent Having the initial conditions, it is easy to see that:
\[
-\q{0}{0} + k \q{(1)}{0} + \q{(2)}{0} - k \q{(3)}{0}  = 0,
\]

{\noindent Substituting the initial conditions into \eqref{eq:Delta:corr}, and using \eqref{eq:corr:monotone}, we have $\E{\Delta_{t^*}} \geq 0$.}

%\begin{align*}
%& \E{\Delta_{t^*}} \geq \\
%& ~~\frac{(k-1)c}{k} \left\{\frac{2}{n}  \sum_{s = 0}^{t^*-1} (1- \frac{2}{n})^{t^*-s-1} \left[ \frac{1}{k} \left(\q{(1)}{s} - \q{(4)}{s}\right) + (k-2) \left(\q{(2)}{s} - \q{(4)}{s}\right)+ \frac{k-1}{k} \left( \q{(3)}{s} - \q{(4)}{s}\right) \right]\right\} \geq 0
%\end{align*}

To complete the proof of Claim \ref{thm:clm}, we need to show that $\E{\Delta_{t^{*}}} \leq 1$. This follows directly form the definition of
$\E{\Delta_{t^*}}$: the number of cooperators cannot increase by more than one after one step. Therefore, the expected drift at each time step is at most one.

%% file: uppr.tex
\section{Proof of Lemma \ref{lem:uppr}}
\label{sec:uppr}
First we study the convergence time of our Markov chain. Let $\theta$ be the convergence time, i.e., $\theta$ is the earliest time that the system reaches either the all one ($\vec{X} = \vec{1}$) or the all zero ($\vec{X} = \vec{0}$) state.
Note that our Markov chain is essentially a random walk on the cube $\{0,1\}^{n}$. In order to compute $\theta$ we can consider a
slightly different random walk in which there is non-zero transition probabilities form state $\vec{1}$ (and $\vec{0}$) to its neighbors in the cube. It is easy to see that starting from any state in $\{0,1\}^{n} \setminus \{\vec{0}, \vec{1}\}$, the time needed to hit $\vec{0}$ or $\vec{1}$ is the same in these two chains. However, the modified chain is irreducible which makes the analysis easier. Let $T_{\vec{X}}$ be time needed to hit $\vec{0}$ or $\vec{1}$ starting form $\vec{X}$ and let $T = \max_{\vec{x} \in \{0,1\}^{n} \setminus \{\vec{0}, \vec{1}\}}{T_{\vec{x}}}$. Clearly $\theta \leq T$.

Next, we proceed to compute $\E{T}$. We show that for $\epsilon < n^{-4}$, the hitting time of WIS, $\E{T} \leq (k+1) n^2/4$.
Note that for $\epsilon = 0$, our dynamics is the voter model. The hitting time of the voter model on a connected
$k$-regular graph is upper-bounded by $k n^2/4$ \cite[Chapter 14, Proposition 9]{aldous_fill}. In the following, we show
that for $\epsilon < n^{-4}$, the expected convergence time, $\E{T}$, is within an additive $O(1)$ factor of the hitting
time of the basic voter model, which implies that $\E{T} \leq (k+1) n^2/4$.

Define the mapping $\mathcal{M}: \{0,1\}^n \rightarrow \{0,1,\ldots,n\}$, where $\mathcal{M}(\vec{X}) = \sum_{i =1}^{n} X^i$. In words, we contract all the states with the same number of cooperators. Note that if we apply this mapping to the modified irreducible Markov chain, the resulting chain will be a {\em reversible death-birth chain}.
The following results for the reversible death-birth chain enable us to show that for $\epsilon$ sufficiently small, the hitting time of WIS is almost the same as the hitting time for the basics voter model.

\vspace{5mm}
\subsection{Review of Some Results for Reversible Death-Birth Chains:}
\label{subsec:DBChain}

\vspace{3mm}
{\noindent Let $\pi_i$, $0 \leq i \leq n$, be the stationary distribution of the death-birth chain, $p_{j-1,j}$ be the transition probability from state $j-1$ to state $j$, $q_{j,j-1}$  be the transition probability from state $j$ to state $j-1$, and $\mathbb{E}_{j-1}\left[T_j\right]$ be the expected time to hit state $j$, starting from state $j-1$.}

\begin{enumerate}
\item Using the reversibility of the chain, we can easily compute the stationary distribution $\pi_i$, $0 \leq i \leq n$ in terms of the
transition probabilities. In particular, given $\pi_{0}$,
\[
\pi_{i} = \pi_{0} \prod_{z = 0}^{j-1} \frac{p_{z,z+1}}{q_{z+1,z}},
\]
we can compute $\pi_0$ by setting $\sum_{i =0}^{n} \pi_{i} =1$.
\item Using the result of Palacios and Tetali \cite{Tetali}, one can explicitly compute all the hitting times of such a chain.
In particular, Theorem $2.3$ of \cite{Tetali} states that:
\begin{align}
\mathbb{E}_{j-1}\left[T_j\right] & = \frac{1}{\pi_{j} q_{j,j-1}} \sum_{z = 0}^{j-1} \pi_{z} \nonumber \\
\mathbb{E}_{j}\left[T_{j-1}\right] & = \frac{1}{\pi_{j} q_{j,j-1}} \sum_{z = j}^{n} \pi_{z}.
\label{eq:lem:tetali}
\end{align}
Since the graph is a line, for any two states $a$ and $b$, where $0 \leq a < b \leq n$:
\begin{align}
\mathbb{E}_{a}\left[T_b\right] = \sum_{j = a}^{b-1}  \mathbb{E}_{j}\left[T_{j+1}\right]
\label{eq:lem:line}
\end{align}
\end{enumerate}

Having the above results for the reversible death-birth chain, we show that for $\epsilon$ sufficiently small, the hitting time of WIS is almost the same as the hitting time of the basics voter model. The difference between each transition probability in the WIS and the voter model is at most $\epsilon$, therefore, the difference in stationary state probability of each state is at most $O(n^2 \epsilon)$.
From relations \eqref{eq:lem:tetali} and \eqref{eq:lem:line}, it is evident that the difference in the hitting times in WIS and the voter model is at most $O(n^4 \epsilon)$, which is $O(1)$ for $\epsilon < n^{-4}$.

{\noindent Finally, note that we have chosen $T^{*}$ such that $T^{*} \geq 2 n^{\gamma/3} \E{T}$. We upper-bound the probability of event $\overline{\mathcal{C}}_{T^{*}}$ in the following way. By Markov inequality,}
\[
\P{}{T \geq 2 \E{T}} \leq \frac{1}{2}
\]
Since the modified chain is irreducible, it satisfies the memoryless property and we have:
\[
\P{}{T \geq 4 \E{T} | \{T \geq 2 \E{T}\}} = \P{}{T \geq 2 \E{T}} \leq \frac{1}{2}
\]
Repeating this argument results in:
\begin{align*}
\P{}{T \geq 2i \E{T}} \leq \frac{1}{2^{i}}, ~~~ i \geq 1
\end{align*}
Thus for $T^{*} \geq 2 n^{\gamma/3}  \E{T}$, we have,
\[
\P{}{\overline{\mathcal{C}}_{T^{*}}} \leq \frac{1}{2^{n^{\gamma/3}}}
\]
which completes the proof of the lemma.

%Based on the above inequality, we derive the following upper bound on $\E{T^2}$ in terms of $\E{T}^2$:
%\begin{align}
%\E{T^2} & = \E{T^2 \Big| \left\{\frac{T^2}{\E{T}^2} \leq 4 \right\}} \P{}{\frac{T^2}{\E{T}^2} \leq 4 } \nonumber \\
%& + \sum_{i = 1}^{\infty} \E{T^2 \Big| \left\{4 i^2  \leq \frac{T^2}{\E{T}^2} \leq 4 (i+1)^2 \right\}} \nonumber \\
%& \times \P{}{4 i^2   \leq \frac{T^2}{\E{T}^2} \leq 4 (i+1)^2 }  \nonumber \\
%& \leq 4 \E{T}^2 \left( 1 + \sum_{i =1}^{\infty} \frac{(i+1)^2}{2^{i}}\right) = 48 \E{T}^2
%\label{eq:lem:upper}
%\end{align}

%\begin{align}
%\E{T^2} & = \sum_{i = 0}^{\infty} \E{T^2 \I{4 i^2  \leq \frac{T^2}{\E{T}^2} \leq 4 (i+1)^2 }} \nonumber \\
%& \leq 4 \E{T}^2 \left( 1 + \sum_{i =1}^{\infty} \frac{(i+1)^2}{2^{i}}\right) = 48 \E{T}^2
%\label{eq:lem:upper}
%\end{align}

%% file: lowr.tex
\section{Proof of Lemma \ref{lem:lowr}}
\label{sec:lowr}
The proof of Lemma \ref{lem:lowr} is mainly algebraic and it amounts for obtaining a lower bound for the solution of a linear dynamical system. We write \eqref{eq:Delta:Ave} in the following matrix form:
\begin{align*}
\E{\Delta_{t^*}} = (k-1) c \vec{Y}^{T} \vec{Q}_{t^{*}},
\end{align*}
where
\begin{align*}
\vec{Y} =
\begin{bmatrix}
-1 \\[0.3em]
\frac{(k-1) b}{kc } \\[0.3em]
1 \\[0.3em]
- \frac{(k-1) b}{kc }
\end{bmatrix}
~~~~~\tm{and} ~~~~
\vec{Q}_{t^*} =
\begin{bmatrix}
\q{(0)}{t^*} \\[0.3em]
\q{(1)}{t^*} \\[0.3em]
\q{(2)}{t^*} \\[0.3em]
\q{(3)}{t^*}
\end{bmatrix}
\end{align*}
Using the recursive relation \eqref{eq:corr:ave}, we have:
\begin{align}
& \E{\Delta_{t^*}} = (k-1) c \vec{Y}^{T} \vec{Q}_{t^{*}} = \nonumber \\
& (k-1) c \vec{Y}^{T} \left[ A^{t^*} \vec{Q}_{0} + \frac{2(k-1)}{nk} \sum_{s = 0}^{t^*-1} A^{t^*-s-1} \vec{R}_{s}\right],
\label{eq:lem:matrix}
\end{align}
where matrix $A$ is:
\begin{align*}
A =
\begin{bmatrix}
       1            & 0               & 0                & 0                 \\[0.3em]
       \frac{2}{nk} & 1 - \frac{2}{n} & \frac{2(k-1)}{nk}& 0                 \\[0.3em]
       0            & \frac{2}{nk}    & 1 - \frac{2}{n}  & \frac{2(k-1)}{nk} \\[0.3em]
       0            & 0               & \frac{2}{nk}     & 1 - \frac{2}{n}    \\[0.3em]
     \end{bmatrix}
\end{align*}
and vector $\vec{R}_{s}$ is
\begin{align*}
\vec{R}_{s} = \q{(4)}{s}
{\begin{bmatrix}
0 & 0 & 0 & 1
\end{bmatrix}}^{T}
\end{align*}
As mentioned before, for a graph with girth not smaller than $7$, $\q{(4)}{s} \leq \q{(3)}{s}$, for $s \leq t^*$. Using this inequality and some tedious algebra, we show that:
\begin{align*}
 \vec{Y}^{T} A^{t^*-s-1}
\vec{R}_{s}\geq
\vec{Y}^{T} A^{t^*-s -1}
\begin{bmatrix}
0 \\[0.3em]
0 \\[0.3em]
- \q{(3)}{s}  \sqrt{(k-1)/2} \\[0.3em]
\q{(3)}{s}
\end{bmatrix}
\end{align*}
Substituting this inequality in \eqref{eq:lem:matrix}, we have:
\begin{align}
&\E{\Delta_{t^*}} \nonumber \\
&\geq (k-1) c \vec{Y}^{T} \left[ A^{t^*} \vec{Q}_{0} + \frac{2(k-1)}{nk} \sum_{s = 0}^{t^*-1} A^{t^*-s-1} \vec{W}_{s}\right],
\label{eq:lem:matrix2}
\end{align}
where vector $\vec{W}_s$ is:
\begin{align*}
\vec{W}_s =
\q{(3)}{s}
\begin{bmatrix}
0 \\[0.3em]
0 \\[0.3em]
- \sqrt{(k-1)/2} \\[0.3em]
1
\end{bmatrix}
\end{align*}
We can  rewrite \eqref{eq:lem:matrix2} as:
\begin{align*}
\E{\Delta_{t^*}} \geq (k-1) c \vec{Y}^{T} B^{t^*} \vec{Q}_{0} ,
\end{align*}
where matrix $B$ is given by,
\begin{align*}
B =
\begin{bmatrix}
       1            & 0               & 0                & 0                 \\[0.3em]
       \frac{2}{nk} & 1 - \frac{2}{n} & \frac{2(k-1)}{nk}& 0                 \\[0.3em]
       0            & \frac{2}{nk}    & 1 - \frac{2}{n}  & \frac{2(k-1)}{nk} (1 - \sqrt{\frac{k-1}{2}}) \\[0.3em]
       0            & 0               & \frac{2}{nk}     & 1 - \frac{2}{nk}    \\[0.3em]
     \end{bmatrix}
\end{align*}
Summing over $0 \leq t^* \leq T^{*}-1$,
\begin{align*}
\sum_{t^* =0}^{T^{*}-1} \E{\Delta_{t^*}} \geq (k-1) c \vec{Y}^{T} \left[ \sum_{t^* =0}^{T^{*}-1}  B^{t^*} \right] \vec{Q}_{0} ,
\label{eq:lem:matrix3}
\end{align*}
The rest of the proof is algebraic; we compute the eigenvalue decomposition of matrix $B$ and based on that, we establish the lower bound on
$\sum_{t^* =0}^{T^{*}-1} \E{\Delta_{t^*}}$.

%% file: example.tex
\section{WIS on Graphs with Small Girth}
\label{sec:example}
In this section, we present two simple examples, complete graph and complete bipartite graph, that show why
the expected number of cooperators does not
increase in graphs with small girth. Because of the symmetry, we can compute the exact drift in
these two examples. The main observation is that in these graphs, the local clustering of cooperators
does not occur. For instance, in the complete bipartite graph, the set of neighbors of all the
nodes in one side of the graph is always the same. Thus we have the free-riding effect that
a defector receives a higher payoff, which results in the decrease in the probability of
imitating cooperation.

\begin{proposition}
Suppose graph $G$ is the complete graph and suppose at time $0$, $x$ nodes play $C$ and the rest of the nodes play $D$; the system evolves according to the $WIS$ afterwards. For any $b,c \geq 0$ and $\epsilon \geq 0$, we have $\pi_{WIS} \leq 2/n$.
\end{proposition}
\begin{proof}
Suppose at time $t$, the configuration $\vec{X}_{t}$ has $y$ cooperators. Using Relation \eqref{eq:Delta:sum:config}, a simple counting shows that,
\[
\E{\Delta_t | ~\vec{X}_{t}} =  -\frac{n-2}{n(n-1)^2} [b + (n-1) c] y(n - y)
 \]
that is non-positive for any $0 \leq y \leq n$. Thus $\E{N_{t}} \leq \E{N_{t-1}}$, for all $t \geq 0$, which immediately implies that $\pi_{WIS} \leq 2/n$.
\end{proof}

\begin{proposition}
Suppose graph $G$ is the complete bipartite graph and $n$ is even. Also, suppose at time $0$, $x$ nodes play $C$ and the rest of the nodes play $D$; the system evolves according to the $WIS$ afterwards. For any $b,c \geq 0$ and $\epsilon \geq 0$, we have $\pi_{WIS} \leq 2/n$
\end{proposition}
\begin{proof}
Suppose at time $t$, the configuration $\vec{X}_{t}$ has $y_1$ and $y_2$ cooperators in the two sides of the graph. Using Relation \eqref{eq:Delta:sum:config}, a simple counting shows that,
\begin{align*}
\E{\Delta_t |~\vec{X}_{t}}  = &- \frac{c}{n} \left[(n/2 - y_1+y_2)y_1 (n/2 - y_2) \right. \\
& \quad \quad \quad  + \left. (n/2-y_2 + y_1) y_2 (n/2 - y_1) \right]
\end{align*}
that is non-positive for any  $0 \leq y_1,y_2 \leq n/2$. Thus $\E{N_{t}} \leq \E{N_{t-1}}$, for all $t \geq 0$, which immediately implies that $\pi_{WIS} \leq 2/n$.
\end{proof} 

%% file: discussion.tex
\section{Discussion}
\label{sec:discuss}
\subsection{Relations to Evolutionary Dynamics of Nowak et al.}
\label{subsec:nowakdynamic}
As mentioned in Section \ref{sec:intro}, our WIS dynamics is closely related to a dynamics of Nowak et al. \cite{nowak_lieberman,nowak}.
In their work, they use the following updating rule: node $i$ updates its action to $C$ with a probability proportional to the fitness of its cooperator neighbors. Fitness of node $j$ is defined to be  $1- \epsilon + \epsilon \U{j}{t}$, where $\epsilon$ is a small positive number. In particular,
\begin{eqnarray}
\Pt{\x{i}{t+1} = 1|\mathcal{A}_{t}^{i}, \vec{X}_t} = \frac{\sum_{j \sim i}[(1-\epsilon)+ \epsilon \U{j}{t}] \x{j}{t} }{k(1-\epsilon) + \epsilon \sum_{j \sim i} \U{j}{t}}, \nonumber \\
\label{eq:nowak_dynamic}
\end{eqnarray}

This dynamics models the mechanism that node $i$ dies and its cooperator (defector) neighbors compete to replace $i$ with a cooperator (defector) player in proportion to their fitness. When $\epsilon \ll 1$, the effect of the payoffs is quite small and the process is called Weak Selection. They show that in the Weak Selection regime, if $b > kc$, the expected number of cooperators increases in the long run of the process. However, their analysis lacks rigor and ignores the effect of cycles and the correlation between nodes at distance more than one.

As it is evident from relation \eqref{eq:nowak_dynamic}, this is a nonlinear dynamics and its rigorous analysis is prohibitively difficult. However, note that algebraic manipulation results in having:

\begin{eqnarray*}
\Pt{\x{i}{t+1} = 1|\mathcal{A}_{t}^{i}, \vec{X}_t} = \frac{1}{k} \sum_{j \sim i}  \x{j}{t}  \left[ 1 - \epsilon \right. \\
\left. + \epsilon \left(\U{j}{t}+1 - \frac{1}{k}\sum_{h \sim i} \U{h}{t}\right)\right] + O(\epsilon^2/k)
\end{eqnarray*}
which is the same as the transition kernel of the WIS process (relation \eqref{eq:updat_WIS}) up to an $O(\epsilon^2/k)$ factor.

\subsection{Constant $f$ in Theorem \ref{thm:main}}

We compute the constant $f$ by putting Lemma \ref{lem:uppr} and \ref{lem:lowr} together:
\begin{align*}
f & =  \frac{k (k-1)^3 \sqrt{(k-1)/2}}{k^2(k-1) \sqrt{(k-1)/2} + k} (b - \frac{k^2}{k-1} c) \\ & - 3/4(k+1)^2 n^{-\gamma/3} - \frac{n^{5 + \gamma}}{2^{n^{\gamma/3}}}.
\end{align*}